\documentclass[a4paper,10pt]{llncs}
\usepackage[utf8]{inputenc}
\usepackage[english]{babel}
\usepackage{amsmath,amssymb,amsfonts,graphics,graphicx}

{\itshape}{\rmfamily}

\title{Polynomial-time approximability\\of the $k$-{\sc Sink Location} problem}
\author{R\'emy Belmonte \and Yuya Higashikawa \and Naoki Katoh \and Yoshio Okamoto}

\author{
\mbox{
R\'emy Belmonte$^1$
\hspace{1cm}
Yuya Higashikawa$^1$
}
\mbox{
Naoki Katoh$^1$
\hspace{1cm}
Yoshio Okamoto$^2$
}
}
\institute{
$^1$Department of Architecture and Architectural Engineering,\\
C-2 cluster, 4 Katsura, Nishikyo-ku, Kyoto University, Japan\\
\texttt{remybelmonte@gmail.com, \{as.higashikawa,naoki\}@archi.kyoto-u.ac.jp}\\
\medskip
$^2$Department of Communication Engineering and Informatics,\\
Graduate School of Informatics and Engineering,\\
The University of Electro-Communications\\
\texttt{okamotoy@uec.ac.jp}
}

\begin{document}

\maketitle

\begin{abstract}
A dynamic network ${\cal N} = (G,c,\tau,S)$ where $G=(V,E)$ is a graph, integers $\tau(e)$ and $c(e)$ represent, for each edge $e\in E$, the time required to traverse edge $e$ and its nonnegative capacity, and the set $S\subseteq V$ is a set of sources. In the $k$-{\sc Sink Location} problem, one is given as input a dynamic network ${\cal N}$ where every source $u\in S$ is given a nonnegative supply value $\sigma(u)$. The task is then to find a set of sinks $X = \{x_1,\ldots,x_k\}$ in $G$ that minimizes the routing time of all supply to $X$. Note that, in the case where $G$ is an undirected graph, the optimal position of the sinks in $X$ needs not be at vertices, and can be located along edges. Hoppe and Tardos\cite{HT00} showed that, given an instance of $k$-{\sc Sink Location} and a set of $k$ vertices $X\subseteq V$, one can find an optimal routing scheme of all the supply in $G$ to $X$ in polynomial time, in the case where graph $G$ is directed. Note that when $G$ is directed, this suffices to obtain polynomial-time solvability of the $k$-{\sc Sink Location} problem, since any optimal position will be located at vertices of $G$.  However, the computational complexity of the $k$-{\sc Sink Location} problem on general undirected graphs is still open. In this paper, we show that the $k$-{\sc Sink Location} problem admits a fully polynomial-time approximation scheme (FPTAS) for every fixed $k$, and that the problem is $W[1]$-hard when parameterized by $k$.
\end{abstract}

\section{Introduction}

One of the most well-studied family of problems in algorithmic graph theory is that of so-called network flow problems. Such problems originate from the seemingly simple idea that one wishes to move from one location to another in a network, while minimizing some cost function. From this first idea, a wealth of natural problems were derived, where one usually seeks to transport a set of items from a prescribed set of sources to a non-necessarily predetermined set of target locations. In this paper, we study the problem known as $k$-{\sc Sink Location} on dynamic networks, where one seeks to find a set of locations in a network which minimizes the amount of time required to evacuate all the supply located on each vertex of the graph to these locations. A dynamic network is graph with a prescribed set of sources, and nonnegative capacities and integral transit times for every edge. A similar problem, called {\sc Quickest Transshipment}, was studied by Hoppe and Tardos~\cite{HT00}. In this problem, one is given a set of sinks together with the sources and each sink has a demand value. The task is then to send the supply from the sources to the sinks as quickly as possible, in such a way that each sink receives exactly as much supply as its demand value.  They showed that the problem is polynomial-time solvable in the case where the dynamic network is directed, which can easily be seen to imply polynomial time solvability of the $k$-{\sc Sink Location} problem on directed graphs, for every fixed $k$. Note that polynomial-time solvability for the directed case does not readily imply solvability for the undirected case. This is due to the fact that while in a directed network an optimal solution can only be located at a vertex, an optimal solution in an undirected network may be located at any point along an edge.

{\em Related work.} For the 1-{\sc Sink Location} problem, Kamiyama, Katoh and Takizawa~\cite{KKT06,KKT09} gave efficient algorithms for several cases where the structure of the network satisfies requirements regarding the length of the edges and the structure of the given graph. Mamada et al.~\cite{MUMF06} gave an $O(n\cdot \log^2{n})$ algorithm for the case where the input graph is a tree. Most of the recent work on the {\sc Sink Location} problem has been considering computationally harder variants of the problem on restricted graph classes. In~\cite{HGK14}, Higashikawa, Golin and Katoh considered the generalized version of the problem where one seeks not only to find a single sink, but some given number $k$, when the input graph is an undirected path. They showed that this problem can be solved in $O(kn)$ time in the so-called minimax setting, and $O(n^2 \cdot \min\{k,2^{\sqrt{\log{k} \log\log{n}}}\})$ in the minisum setting. In~\cite{HGK14a}, the same authors considered the minimax  regret version of the problem when the input graph is a tree and the edges have uniform capacities, and showed that the problem can be solved in $O(n^2 \log^2 n)$ time. As an intermediate result, they provide an $O(n \log n)$ algorithm to find a sink location that minimizes the evacuation time.

{\em Some definitions.} We now define formally the notion of dynamic network and the $k$-{\sc Sink Location} problem. A {\em dynamic network} ${\cal N} = (G, c, \tau, S)$ consists of a graph $G=(V,E)$ where each edge $e\in E(G)$ is given a {\em capacity} $c(e)\in\mathbb{N}$ and a {\em transit time} $\tau(e)\in \mathbb{N}^+$, and a prescribed set of {\em sources} $S\in V$. The notion of dynamic network was first introduced by Ford and Fulkerson~\cite{FF62}. In the $k$-{\sc Sink Location} problem, one is given a dynamic network ${\cal N}$ together with a supply value $\sigma(u)\in \mathbb{N}$ for each vertex $u$ of the set $S$ of sources of ${\cal N}$. The task is then to find $k$ positions $X=\{x_1,\ldots,x_k\}$ in the graph, called {\em sinks}, which minimizes the amount of time required to send the supply $\sigma(u)$ from each vertex $u\in S$ to the positions $X$. A position $x$ is defined as a triple $(uv,\tau(ux),\tau(vx))$, where $e$ is an edge of $G$, and $\tau(pq)\in\mathbb{N}$ represent the time required to travel from position $p$ to $q$, and $\tau(ux)+\tau(xv)=\tau(uv)$. Observe that the number of positions for every edge $e$ is exactly $\tau(e)+1$, and can therefore be exponentially large in the size of the input graph $G$.
In the {\sc Quickest Transshipment} problem, we are also given a function $\sigma: S \rightarrow \mathbb{Z}\setminus \{0\}$. For every vertex $s\in S$, if $\sigma(s) > 0$ then $s$ is called a {\em source}, otherwise it is called a {\em sink}. The question is then to send all the supply from sources to the sinks in a minimum amount of time, in such a way that each sink $s$ receives exactly $-\sigma(s)$ units of supply. Note that this immediately implies $\sum_{s\in S}\sigma(s)=0$. Hoppe and Tardos~\cite{HT00} proved that {\sc Quickest Transshipment} can be solved in polynomial time on directed graphs.
For additional terminology and notation, refer to the monograph by Diestel~\cite{Die05}.

{\em Our contribution.} We study the computational complexity of the $k$-{\sc Sink Location} problem on general undirected graphs and prove the following result:

\begin{theorem}
\label{thm:FPTAS}
The $k$-{\sc Sink Location} problem admits an FPTAS on undirected dynamic networks for every fixed $k$. 
\end{theorem}

A parameterized problem is said to be FPT  by some parameter $k$ if there is an algorithm that solves the problem in time $f(k)\cdot n^{O(1)}$. Intuitively, a $W[1]$-hard problem is a problem that is unlikely to admit an FPT algorithm. We refer the reader to~\cite{FG06,Nie06} for more information about parameterized complexity and algorithms.
We complement Theorem~\ref{thm:FPTAS} by showing that is unlikely to be significantly improved:

\begin{theorem}
\label{thm:W-hard}
The $k$-{\sc Sink Location} problem is $W[1]$-hard when parameterized by $k$. 
\end{theorem}

\section{Polynomial-time approximation scheme}

In this section, we prove our main result, namely that the $k$-{\sc Sink Location} problem admits an FPTAS on general undirected graphs, for every fixed $k$. To that end, we will first need to show that, given an instance $({\cal N}, \sigma)$ of the $k$-{\sc Sink Location} problem and a set of $k$ positions $X$ in $G$, one can compute the minimum amount of time required to send all the supply to $X$. Note the following result of Hoppe and Tardos~\cite{HT00}:
\begin{theorem}[\cite{HT00}]
The {\sc Quickest Transshipment} problem can be solved in polynomial time on directed graphs.
\end{theorem}
Recall that in the {\sc Quickest Transshipment} problem, one is given the set of sinks $X\subseteq V(G)$, and each sink $x\in X$ is given a demand value that must be met exactly. We show that the $k$-{\sc Sink Location} problem on undirected graphs can be reduced to the {\sc Quickest Transshipment} problem on directed graphs when the set of sinks is given. This will later allow us to use Hoppe and Tardos' algorithm to evaluate the time required to send all the supply to a given set of sinks.

\begin{lemma}
\label{lem:reduction}
There is an algorithm that, given an instance $({\cal N}, \sigma)$ of the $k$-{\sc Sink Location} problem where $G$ is undirected and a set of $k$ positions $X$ in $G$ such that no two positions in $X$ lie on the same edge, computes the minimum amount of time required to send all supply to $X$ and runs in polynomial time.
\end{lemma}

\begin{proof}
We reduce our problem to {\sc Quickest Transshipment} on directed graphs. Since Hoppe and Tardos~\cite{HT00} proved that the problem is polynomial-time solvable in that case, this will immediately imply our lemma. Given an instance $({\cal N}, \sigma)$ of the $k$-{\sc Sink Location} problem with ${\cal N}=(G,c,\tau,S)$, where $G$ is undirected, and a set of positions $X=\{x_1,\ldots,x_k\}$ such that $x_i=(u_iv_i,\tau_i^u,\tau_i^v)$ in $G$, we create an instance $({\cal N}', \sigma')$ of the {\sc Quickest Transshipment} problem with ${\cal N}'=(G',c',\tau',S')$, where $G'=(V',E')$ is directed. Our construction is as follows:
\begin{itemize}
\item $V'=V \cup X \cup \{s^*\}$, with $X=\{x_1,\ldots,x_k\}$;
\item $S'=S\cup \{s^*\}$;
\item For every edge $ww'\in E\setminus \bigcup_{i=1}^k\{u_iv_i\}$, we create~2 new opposite edges $ww'$ and $w'w$ such that $c'(ww')=c'(w'w)=c(ww')$ and $\tau'(ww')=\tau'(w'w)=\tau(ww')$;
\item For every edge $uv$ such that there is a position $x\in X$ that lies on $uv$, we replace $uv$ with~4 edges $ux,xu,vx,xv$ such that $c'(ux)=c'(xu)=c'(vx)=c'(xv)=c(uv)$ and $\tau'(ux)=\tau'(xu)=\tau_u$ and $\tau'(vx)=\tau'(xv)=\tau_v$;
\item We add edges $x_is^*$ for every $1\leq i\leq k$ and set $c(x_is^*)=\sum_{s\in S}\sigma(s)$ and $\tau(x_is^*)=0$;
\item For every vertex $w\in S, \sigma'(w)=\sigma(w)$, $\sigma'(x_i)=0$ for every $1\leq i\leq k$ and $\sigma'(x^*)=-\sum_{s\in S}\sigma(s)$.
\end{itemize}
We now claim that the minimum amount of time required to send all supplies to $x$ in $({\cal N}, \sigma)$ is equal to the minimum amount of time required in $({\cal N}', \sigma')$.
The fact that any solution in $({\cal N}, \sigma)$ corresponds to an equivalent solution in $({\cal N}', \sigma)'$ follows from the fact that in $({\cal N}, \sigma)$, an edge is never traversed in both directions at a given time, and the length and capacity of every edge is the same as in $({\cal N}', \sigma)'$. Similarly, for every routing scheme in $({\cal N}', \sigma')$, there exists an equivalent routing that can be completed in the same amount of time where, at any given time, at most one edge out of every pair of opposite edges has non-zero flow in the routing. This concludes the proof of the lemma.
\qed
\end{proof}

We are now ready to describe our FPTAS for the $k$-{\sc Sink Location} problem on undirected graphs. Roughly speaking, we will ``guess'' near-optimal positions for the $k$ sinks by performing sampling at regular interval over each edge of $G$. One can then reduce the approximation ratio by increasing the amount of sampling points. In this section, given an instance of the $k$-{\sc Sink Location} problem on undirected graphs, we denote by $OPT(X)$ the minimum amount of time required to send all supplies to the set of positions $X$, and $OPT=\min\{OPT(X)\mid X=\{x_i=(u_iv_i,\tau(u_ix_i,\tau(x_iv_i))\}$, for all $1\leq i\leq k$, such that  $u_iv_i\in E$ and $\tau(u_ix_i)+\tau(x_iv_i)=\tau(u_iv_i)$.

\begin{proof}[of Theorem~\ref{thm:FPTAS}]
We describe an algorithm that takes as input an instance $({\cal N},\sigma)$ of the $k$-{\sc Sink Location} problem and $\varepsilon>0$ and returns a set of $k$ positions $X$ in $G$ such that $OPT(X) \leq (1+\varepsilon)\cdot OPT$. Let us define $t_e=\max\{1, \left\lfloor\varepsilon\cdot \tau(e)\rfloor\right.\}, \forall e\in E$. 
We first define a set ${\cal X}$ of positions in the following way: ${\cal X}$ consists of all the vertices of $G$, together with positions $X_{uv}=\{x_1.\ldots,x_{\ell_{uv}}\}, \forall uv\in E$, with  $\ell_{uv}=\min\{\tau(uv), \left\lceil \frac{1}{\varepsilon} \right\rceil\}$. Observe that $|{\cal X}| \leq |V| + \frac{|E|}{\varepsilon}$. Our algorithm then tries every possible set $X$ of $k$ positions in ${\cal X}$, computes $OPT(X)$ using Lemma~\ref{lem:reduction}, and returns $\min\{OPT(X) \mid X\subseteq {\cal X} \wedge |X|=k\}$. Our algorithm runs in time ${|{\cal X}| \choose k} \cdot H(n)$, where $H(n)$ is the running time of Hoppe and Tardos' algorithm. The running time of our algorithm is then $O((|V|+\frac{|E|}{\varepsilon})^k \cdot H(n))$, as desired.
To complete the proof, it only remains to show that there exists a set $X$ of $k$ vertices in ${\cal X}$ such that $OPT(X)\leq(1+\varepsilon)\cdot OPT$.
Consider a set of $k$ positions $X^*$ in $G$ such that sending all the supply in $G$ to positions in $X^*$ takes time $OPT$, i.e., $X^*$ is an optimal set of positions in $({\cal N},\sigma)$. First, we show that we may assume, without loss of generality, that at most~2 positions in $X^*$ lie on the same edge $uv$, and if exactly~2 positions of $X^*$ lie on $uv$, then these positions are exactly $u$ and $v$. Indeed, observe first that we may safely assume that no edge $uv$ contains more than~2 positions of $X^*$, since every unit of supply that is sent to a sink on $uv$ will have to pass through either the position $x$ in $X^*$ closest from $u$, or the position $y$ closest from $v$. Note that $u=x$ and $v=y$ may happen. Therefore, we may remove all positions of $X^*$ that lie on $uv$ other than $x$ and $y$. Additionally, since every unit of supply sent to $x$ and $y$ have to pass through either $u$ or $v$ in order to reach a sink, we may safely replace $x$ and $y$ with $u$ and $v$ in $X^*$, without increasing the total amount of time required to send the supply to the sinks.
Consider now the set $X\subseteq {\cal X}$ such that for every position $x\in X^*$, we add $x$ to $X$ if $x\in{\cal X}$, otherwise we add the position $x'$ in ${\cal X}$ closest from $x$. Note that if $x\not\in {\cal X}$, then $x\not\in V$, and therefore $x$ lies on a edge, and lies between two positions of ${\cal X}$. In case $x$ is equidistant from these two positions, we choose one of them arbitrarily. Moreover, since every edge $uv$ contains either at most~1 position of ${\cal X}$ or both $u$ and $v$, no two distinct positions of $X^*$ correspond to the same position $x'$ in $X$, and hence, each position $x\in X^*$ is associated with a position $x'\in X$ in a bijective manner.
We now claim that $X$ satisfies $OPT(X)\leq(1+\varepsilon)\cdot OPT$, as desired.
To prove this claim, we show that any routing to $X^*$ that can be achieved in time $OPT$ can be transformed into a new routing to $X$ that can be achieved in time at most $OPT+\max\{\frac{t_e}{2} \mid e\in E \wedge \exists x\in X, x\in E\}$. This immediately follows from the fact that for every pair of positions $x$ and $x'$, the supply sent to $x$ either passes through $x'$, in which case it can simply stop there, or it can be sent to $x'$ as soon as it reaches $x$, without violating the capacity constraint. Since $x'$ is chosen to be closest from $x$ among the positions in ${\cal X}$, and the distance between two positions of ${\cal X}$ lying on an edge $e$ is at most $t_e$, we obtain that the supply reaches $x'$ in the new routing at most $\frac{t_e}{2}$ units of time after it reaches $x$. Note that, if we denote by $OPT(x)$ the time at which the last unit of supply reaches $x$ in the original routing, and $OPT'(x')$ the time at which the last unit of supply reaches $x'$ in the modified routing, we have for every pair of positions $x,x'$:
\[OPT'(x') \leq OPT(x)+\frac{t_e}{2}\]
Where $e$ is the edge that contains $x$ and $x'$. Moreover, observe that if $x\in X$, then $x'=x$ and $OPT'(x')=OPT(x)$, and if $x'\neq x$ then $t_e = \left\lfloor\varepsilon\cdot \tau(e)\rfloor\right.$.
Therefore, we have
\[OPT'(x')\leq OPT(x) + \frac{\varepsilon\cdot \tau(e)}{2}\]
Finally, observe that if $OPT'(x')\neq OPT(x)$, we way assume without loss of generality that $x\not\in \{u,v\}$, and hence there is at least~1 unit of supply reaching $x$ in the original routing that passes through $u$ and at least~1 unit that passes through $v$. Hence, for every position $x\in X$ lying on edge $e$:
\[OPT(x) \geq \frac{\tau(e)}{2}\]
Which in turn implies
\[OPT'(x') \leq OPT(x) + \frac{\varepsilon \cdot OPT(x)}{2} \leq OPT(x)(1+\varepsilon)\]
Since this inequality holds for every pair of positions $x$ and $x'$, we obtain the following inequality for the global solutions $OPT$ and $OPT'$:
\[OPT'\leq OPT(1+\varepsilon)\]
This concludes the proof of our main theorem.
\qed
\end{proof}

\section{Hardness of $k$-{\sc Sink Location} parameterized by $k$}

In this section, we provide a simple reduction from the well-known $W[2]$-complete problem $k$-{\sc Hitting Set}. In this problem, one is given as input a ground set $U$ and a family of sets ${\cal X}$, and the task is to find a subset $U'$ of $U$ containing at most $k$ elements, such that every set in ${\cal X}$ contains at least~1 element of $U'$.

\begin{theorem}
The $k$-{\sc Sink Location} problem is $W[2]$-hard parameterized by $k$ on undirected graphs, even when all the edges in the input graph $G$ have unit length and capacity.
\end{theorem}

\begin{proof}
Given an instance $(U,{\cal X})$ of {\sc Hitting Set}, we build a graph $G=(V,E)$ such that $V=U\cup{\cal X}$, and two vertices $x\in U$ and $y\in{\cal X}$ are made adjacent whenever  $x\in y$. We then set $c(e)=\tau(e)=1$ for every edge $e\in E$ and add exactly~1 unit of supply to each vertex in ${\cal X}$. We now claim that $(U,{\cal X})$ admits a hitting set of size $k$ if and only if there exist $k$ sinks in $G$ to which all the supply can be sent in exactly~1 unit of time.

For the forward direction, observe that if $(U,{\cal X})$ has a hitting set of size $k$, then choosing those $k$ elements of $U$ as sinks will allow to send all the supply in a single unit of time.

For the converse direction, observe that since the vertices of ${\cal X}$ form an independent and every edge in $E$ has length~1, every sink that lies on an edge $uv$ with $u\in U$, but not at $u$, will only be able to receive supply from $v$. Hence, we may assume that all the sinks lie on vertices of $U$. It is then clear that the $k$ sinks must be adjacent to every vertex of ${\cal X}$, and form therefore a hitting set of $(U,{\cal X})$. 
\qed
\end{proof}

\section{Conclusion}

In this paper, we proved that the $k$-{\sc Sink Location} problem admits an FPTAS on general undirected graphs for every fixed $k$, but that, on the negative side, it is $W[2]$-hard when $k$ is not fixed, but used as a parameter instead. Two natural questions immediately follow. The first of these two questions is whether the $k$-{\sc Sink Location} problem can be solved in polynomial time for every fixed $k$. The second is whether the problem admits an FPTAS when parameterized by $k$, i.e., can be solved in time $f(k)\cdot poly(n,\frac{1}{\varepsilon})$ with approximation ratio $1+\varepsilon$, for some function $k$.

\end{document}